\pdfoutput=1
\documentclass{article}

\usepackage[a4paper, left=2.5cm,right=2.5cm,top=3cm,bottom=3cm]{geometry}
\usepackage[utf8]{inputenc} 
\usepackage[T1]{fontenc}    
\usepackage{url}            
\usepackage{booktabs}      
\usepackage{amsfonts}       
\usepackage{nicefrac}       
\usepackage{microtype}      
\usepackage{lipsum}
\usepackage{physics}
\usepackage{amsmath}
\usepackage{amssymb}
\usepackage{mathtools}               
\usepackage{booktabs}       
\usepackage{amsfonts}           
\usepackage{amsthm}
\usepackage{cancel}
\usepackage[explicit]{titlesec}
\usepackage[square,numbers]{natbib}
\newtheorem{Theorem}{Theorem}
\newtheorem{Definition}{Definition}
\newtheorem{Proposition}{Proposition}
\newtheorem{Lemma}{Lemma}
\newtheorem{Remark}{Remark}
\newtheorem*{claim*}{Claim}
\makeatletter
\newcommand*{\defeq}{\mathrel{\rlap{%
                     \raisebox{0.3ex}{$\m@th\cdot$}}%
                     \raisebox{-0.3ex}{$\m@th\cdot$}}%
                     =}

\makeatother

\title{On the Mathematics of Coframe Formalism and Einstein--Cartan Theory - A Brief Review}

\author{
 Manuel Tecchiolli\\
  Institute for Theoretical Physics, ETH Z\"{u}rich,\\
  Wolfgang-Pauli-Str. 27, 8093, Zürich, Switzerland
}

\date{24 September 2019}

\begin{document}
\maketitle

\abstract{This article is a~review of what could be considered the basic mathematics of Einstein--Cartan theory. We discuss the formalism of principal bundles, principal connections, curvature forms, gauge fields, torsion form, and Bianchi identities, and eventually, we will end up with Einstein--Cartan--Sciama--Kibble field equations and conservation laws in their implicit formulation.}

General Relativity, \and Torsion-Gravity, \and Mathematical Physics.

\tableofcontents

\section{Introduction}
The formulation of torsion gravity and the consequent coupling with spin rely on a~different formulation compared to the one of original works on General Relativity. This formulation regards geometrical objects called principal bundles. In this context, we can formulate General Relativity (or Einstein--Cartan--Sciama--Kibble (ECSK) theory in the presence of torsion) with a~principal connection, which can be pulled back to the base manifold in a~canonical way giving birth to a so called gauge field and consequently to the well-known spin connection. This process shows the possibility of formulating General Relativity as a~proper gauge theory rather than using the affine formulation and Christoffel symbols $\Gamma$. What permits the equivalence of the two formulations is a~{bundle isomorphism} called \textit{tetrads} or \textit{vierbein}, which is supposed to respect certain compatibility conditions. Then, we can define the associated torsion form and postulate the Palatini--Cartan action as a~functional of such tetrads and spin connection. This leads to ECSK field equations.

We will first set up all the abstract tools of principal bundles, tetrads, and principal connection; secondly, we will derive the Einstein--Cartan--Sciama--Kibble theory in its implicit version; and finally, we will discuss conservation laws coming from local SO$(3,1)$ and diffeomorphism invariance of ECSK theory.
\\

Throughout the article, we will give theorems and definitions. However, we would like to stress that hypotheses for such theorems will often be slightly redundant: we will take spaces and functions to be differentiable manifolds and smooth, even though weaker statements would suffice. This is because we prefer displaying the setup for formalizing the theory rather than presenting theorems and definitions with weaker hypotheses that we will never use for the theory. Nontheless, we will sometimes specify where such hypotheses are strengthened. In spite of this, the discussion will be rather general, probably more general than what is usually required in formulating ECSK (Einstein--Cartan--Sciama--Kibble) theory.

\section{Bundle Structure}
\label{sec:bundles}
The introduction of a~metric $g$ and an orthogonality relation via a~minkowskian metric $\eta$ are two fundamental ingredients for building up a~fiber bundle where we want the orthogonal group to act freely and transitively on the fibers. This will allow us to have a~principal connection and to see the perfect analogy with an ordinary gauge theory (\cite{BaezMuniain} chapter III).

Such a~construction underlies the concept of \textit{principal bundle}, and tetrads will be an isomorphism from the tangent bundle\footnote{Disjoint union of tangent spaces: $TM=\cup_{x\in M}\{x\}\times T_xM$} $TM$ to an associated bundle $\mathcal V$.

\subsection{$G$-Principal Bundle}
We give some definitions\footnote{References \cite{KN,WM,morita} are recommended for further details.}.

\begin{Definition}[G-principal bundle\footnote{We give the definition based on our purposes; in general, we can release some hypotheses. In particular, $G$ needs to be only a~locally compact topological group and $M$ needs to be a~topological Hausdorff space. This definition is a~version with a stronger hypothesis than the one contained in Reference \cite{OSI}.}]
Let $M$ be a~differentiable manifold and G be a~Lie group.

A $G$-principal bundle $P$ is a fiber bundle $\pi:P \to M$ together with a~smooth (at least continuous) right action $\mathfrak P: G\times P\to P$  such that $
\mathfrak P$ acts freely and transitively on the fibers\footnote{Fibers are $\pi^{-1}(x)$ $\forall x\in M$.} of $P$ and such that $\pi(\mathfrak P_g(p))=\pi(p)$ for all $g\in G$ and $p\in P$.
\end{Definition}
We need to introduce a~fundamental feature of fiber bundles.
\begin{Definition}[Local trivialization of a~fiber bundle]
Let $E$ be a~fiber bundle over M, a differentiable manifold, with fiber projection $\pi:E\to M$, and let $F$ be a~space\footnote{In the present case, $F$ will be a~differentiable manifold, a~vector space, a~topological space, or a~topological group. Furthermore, if we write ``space'', we mean one among these.}.

\textls[-20]{A local trivialization $(U,\varphi_U)$ of $E$, is a~neighborhood $U\subset M$ of $u\in M$ together with a~local diffeomorphism.}
\begin{equation}
\varphi_U: U\times F\to \pi^{-1}(U)
\end{equation}
such that $\pi(\varphi_U(u,f))=u\in U$ for all $u\in U$ and $f\in F$.
\end{Definition}
This definition implies $\pi^{-1}(u)\simeq F$  $\forall u\in U$.
\begin{Definition}[Local trivialization of a~$G$-principal bundle]
Let $P$ be a~$G$-principal bundle.

\textls[-20]{A local trivialization $(U,\varphi_U)$ of $P$ is a~neighborhood $U\subset M$ of $u\in M$ together with a~local diffeomorphism.}
\begin{equation}
\label{eq:loctriv}
\varphi_U: U\times G\to \pi^{-1}(U)
\end{equation}
such that $\pi(\varphi_U(u,g))=u\in U$ for all $u\in U$ and $g\in G$ and such that
\begin{equation}
\label{eq:commutationofactions}
\varphi_U^{-1}(\mathfrak P_g(p))=\varphi^{-1}_U(p)g=(u,g')g=(u,g'g).
\end{equation}
\end{Definition}
\noindent{\textbf{Observation 1:}} A fiber bundle is said to be \textit{locally trivial} in the sense that it admits a~local trivialization for all $x\in M$, namely there exists an open cover $\{U_i\}$ of $M$ and a~set of diffeomorphisms $\varphi_i$ such that every $\{(U_i,\varphi_i)\}$ is a~local trivialization\footnote{The bundle is said to be \textit{trivial} if there exists $(U,\varphi_U)$ with $U = M$.}.

Here, we recall the similarity with a~differentiable manifold. For a~manifold when we change charts, we have an induced diffeomorphism between the neighborhoods of the two charts, given by the composition of the two maps.

Thus, having two charts $(U_i, \phi_i)$ and $(U_j, \phi_j)$, we define the following:
\begin{equation}
\phi_j\circ\phi_i^{-1}:\phi_i(U_i\cap U_j)\to \phi_j(U_i\cap U_j).\end{equation}

At a~level up, we have an analogous thing when we change trivialization. Of course, here, we have one more element: the element of fiber.

Taking two local trivializations $(U_i, \varphi_i)$ and $(U_j, \varphi_j)$ and given a~smooth left action $\mathcal T:G\to \operatorname{Diffeo}(F)$ of $G$ on $F$, we then have
\begin{equation}
(\varphi_j^{-1}\circ\varphi_i)(x,f)=\big(x,\mathcal T(g_{ij}(x))(f)\big) \hspace{1cm} \forall x\in U_i\cap U_j, f\in F.
\end{equation}

\noindent where the maps $g_{ij}:U_i\cap U_j\to G$ are called the \textit{transition functions} for this change of trivialization and $G$ is called the structure group.

Such functions obey the following transition functions conditions for all $x\in U_i\cap U_j $:
\begin{itemize}
\item[--] $ g_{ii}(x)=id$
\item[--] $ g_{ij}(x)=(g_{ji}(x))^{-1}$
\item[--] $g_{ij}(x)=g_{ik}(x) g_{kj}(x)$ for all $x\in U_i\cap U_k\cap U_j$.
\end{itemize}
The last condition is called the \textit{cocycle condition}.
\begin{Theorem}[Fiber bundle construction theorem]
Let $M$ be a~differentiable manifold, $F$ be a~space, and $G$ be a~Lie group with faithful smooth left action $\mathcal T:G\to \operatorname{Diffeo}(F)$ of $G$ on $F$.

Given an open cover $\{U_i\}$ of $M$ and a~set of smooth maps,
\begin{equation}
t_{ij}:U_i\cap U_j \to G
\end{equation}
defined on each nonempty overlap, satisfying the transition function conditions.

Then, there exists a~fiber bundle $\pi:E\to M$ such that
\begin{itemize}
\item[--] $\pi^{-1}(x)\simeq F$ for all $x\in M$
\item[--] its structure group is $G$, and
\item[--] it is trivializable over $\{U_i\}$ with transition functions given by $t_{ij}$.
\end{itemize}
\end{Theorem}
A proof of the theorem can be found in Reference \cite{Sharpe} (Chapter $1$).

\subsection{Coframe Bundle and Minkowski Bundle}

It is clear now that having $E$ as a~fiber bundle over $M$ with fibers isomorphic to $F$ and $F'$ as a~space equipped with the smooth action $\mathcal T'$ of $G$, implies the possibility of building a~bundle $E'$ \textit{associated} to $E$, which shares the same structure group and the same transition functions $g_{ij}$. By the fiber bundle construction theorem, we have a~new bundle $E'$ over $M$ with fibers isomorphic to $F'$.

This bundle is called the \textit{associated bundle} to $E$.

Depending on the nature of the associated bundle\footnote{We will be dealing with two particular types of associated bundles: a~principal bundle associated to a~vector bundle and a~vector bundle associated to a~principal bundle.}, we have the following two definitions:
\begin{Definition}[Associated $G$-principal bundle]
Let $\pi:E\to M$ be a~fiber bundle over a~differentiable manifold $M$, $G$ be a~Lie group, $F'$ be a~topological space, and $\mathfrak P$ be a smooth right action of $G$ on $F'$. Let also $E'$ be the associated bundle to $E$ with fibers isomorphic to $F'$.

If $F'$ is the principal homogeneous space\footnote{The space where the orbits of $G$ span all the space.} for $\mathfrak P$, namely $\mathfrak P$ acts freely and transitively on $F'$, then $E'$ is called the $G$-principal bundle associated to $E$.
\end{Definition}

\begin{Definition}[Associated bundle to a~$G$-principal bundle]
Let $P$ be a~$G$-principal bundle over M, $F'$ be a~space, and $\rho$: G$\to$ $\operatorname{Diffeo}(F')$ be a~smooth effective left action of the group $G$ on $F'$.

We then have an induced right action of the group $G$ over $P\times F'$ given by
\begin{equation}
(p,f')*g=(\mathfrak P_g(p),\rho(g^{-1})(f')).
\end{equation}

We define the associated bundle $E$ to the principal bundle $P$, as an equivalence relation:
\begin{equation}
E\coloneqq P\times_\rho F'=\frac{P\times F'}{\sim},
\end{equation}
where $(p,f')\sim(\mathfrak P_g(p),\rho(g^{-1})(f'))$, $p\in P$, and $ f'\in F'$ with projection $\pi_\rho:E\to M$ s.t. $\pi_\rho([p,f'])=\pi(p)=x\in M$. 
\end{Definition}

Therefore $\pi_\rho:E\to M$ is a~fiber bundle over $M$ with $\pi^{-1}_\rho(x)\simeq F'$ for all $x\in M$.

\vspace{12pt}
\noindent{\textbf{Observation 2:}} The new bundle, given by the latter definition, is what we expected from a~general associated bundle: a~bundle with the same base space, different fibers, and the same structure group.

\vspace{12pt}
\noindent{\textbf{Idea:}} We take a~$G$-principal bundle $P$ as an associated bundle to $TM$, and we build a~vector bundle associated to $P$ with a~fiber-wise metric $\eta$. We shall call this associated bundle $\mathcal V$.

\vspace{12pt}
First of all, we display the $G$-principal bundle as the $G$-principal bundle associated to $TM$.

\begin{Definition}[Orthonormal coframe]
Let $(M,g)$ be a~pseudo-riemannian $n$-dimensional differentiable manifold and $(V,\eta)$ be an $n$-dimensional vector space with minkowskian metric $\eta$. \\
A coframe at $x\in M$ is the linear isometry.
\begin{equation}
{}_xe\coloneqq\big\{{}_xe:T_xM\to V\big | \hspace{0.15cm} {}_xe^*\eta\coloneqq\eta_{ab}\,{}_x e^a{}_xe^b=g\big\},
\end{equation}
equivalently ${}_xe^a$ forms an ordered orthonormal basis in $T^*_xM$.

An orthonormal frame is defined as the dual of a~coframe.
\end{Definition}

\noindent{\textbf{Observation 3:}} Locally, coframes can be identified with local covector fields. A necessary and sufficient condition for identifying them with global covector fields (namely a~coframe for each point of the manifold) is to have a~\textit{parallelizable} manifold, namely a~trivial tangent bundle.

\begin{Definition}[Orthonormal coframe bundle]
Let $(M,g)$ be a~differentiable $n$-dimensional manifold with pseudo-riemannian metric $g$ and $T^*M$ be its cotangent bundle (real vector bundle of rank $n$).

We call the coframe bundle $F^*_O(M)$ the $G$-principal bundle where the fiber at $x\in M$ is the set of all orthonormal coframes at $x$ and where the group $G=\operatorname{O}(n-1,1)$ acts freely and transitively on them.
\end{Definition}

The dual bundle of this is the orthonormal frame bundle, and it is denoted by $F_O(M)$, made up of orthonormal frames (dual of orthonormal coframes).

\vspace{12pt}
\noindent{\textbf{Observations 4:}} \begin{enumerate}
\item The orthonormal frame bundle is an associated $G$-principal bundle to $TM$.
\item We can consider the \textit{Minkowski bundle} $\mathcal V$ the vector bundle over $M$ with fibers $V$. It is clear that such a bundle and $F_O(M)$ are one of the associated bundles of the other via action of the orthogonal group $\operatorname{O}(n-1,1)$. Therefore, $\mathcal V\coloneqq F_O(M) \times_\rho V$, where $\rho$ is taken to be the fundamental representation of $O(n-1,1)$.
\item We stress that this bundle $\mathcal V$ is not canonically isomorphic to $TM$; in general, there is no canonical choice of a~representative of ${}_xe$ of the equivalence class $[{}_xe,v]\in\mathcal V$, of which the inverse ${}_xe^{-1}(v)$ gives rise to a~canonical identification of a~vector in $T_xM$. Namely, fixed a~$v\in V$, not all choices of ${}_xe$ give rise to a~fixed vector $X\in T_xM$. As a~matter of fact, the reference metric fixed on $V$ does not allow in general the existence of a~canonical soldering (Section \ref{sec7}).  In Reference \cite{Gielen:2012fz}, it is shown how to define the Minkowski bundle without deriving it from $F_O(M)$; the authors refer to that as \textit{fake tangent~bundles}.
\item If~the manifold is parallelizable, we have the bundle isomorphism $e:TM\to \mathcal V$, which is given by the identity map over $M$ and ${}_xe:T_xM\to V$ $\forall x\in M$. It can be regarded as a~$\mathcal V$-valued $1$-form $e\in\Omega^1(M, \mathcal V)$. We can identify $e$ with an element of $\Omega^1(M,V)$, thus with global sections of the cotangent bundle such that, at each point in $M$, the corresponding covectors ${}_xe^a$ obey $\eta_{ab}\,{}_xe^a{}_xe^b=g$.
\end{enumerate}

We are now ready to define tetrads.

\begin{Definition}[Tetrads]
Let $\rho:O(3,1)\to\operatorname{Aut}(V)$ be the fundamental representation.

Tetrads are the bundle isomorphisms $e:TM\to\mathcal V$. They are identifiable with elements $e\in\Omega^1(M, \mathcal V)$, and if $M$ is parallelizable, tetrads can be identified with $\Omega^1(M, V)\ni e^av_a$ such that $\{v_a\}$ is an orthonormal basis of $V$, $e^a\in\Omega^1(M)$, and $\eta_{ab}e^ae^b=g$.
\end{Definition}

\section{Principal Connection}
\vspace{0.15cm}
\noindent\textit{Is there any difference?}
\vspace{0.3cm}\\
In the ordinary formulation of General Relativity (as in the original Einstein's work, for instance), we have objects called $\Gamma s $, which are coefficients of a~linear connection $\nabla$ and thus determined by a~parallel transport of tangent vectors.

The biggest advantage of treating $\operatorname{O}(3,1)$ as an ``explicit symmetry'' of the theory is that we have obtained the possibility of defining a~\textit{principal connection}, which is the same kind of entity we have in an ordinary gauge theory\footnote{Think of $(U(1),A_\mu)$ for electromagnetism.}.

\subsection{Ehresmann Connection}
If we consider a~smooth fiber bundle $\pi:E\to M$, where fibers are differentiable manifolds, we~can of course take tangent spaces at points $e\in E$. Having the tangent bundle $TE$, we may wonder if it is possible to separate the contributions coming from $M$ to the ones from the fibers.

This cannot be done just by stating $TE=TM\oplus TF$, unless $E=M\times F$ is the trivial bundle. Namely, we cannot split directly vector fields on $M$ from vector fields on the fibers $F$.

We can formalize this idea: use our projection $\pi$ for constructing a~tangent map $\pi_*=d\pi:TE\to TM$, and consider its kernel.

\begin{Definition}[Vertical bundle]
Let $M$ be a~differentiable manifold and $\pi:E\to M$ be a~smooth fiber bundle.

We call the sub-bundle $VE=\operatorname{Ker}(\pi_*:TE\to TM)$ the vertical bundle.
\end{Definition}

Following this definition, we have the natural extension to the complementary bundle of the vertical bundle, which is somehow the formalization of the idea we had of a~bundle that takes care of tangent vector fields on $M$.

\begin{Definition}[Ehresmann connection]
\textls[-20]{Let $M$ be a~differentiable manifold and $\pi:E\to M$ be a~smooth fiber~bundle.}

Consider a~complementary bundle $HE$ such that $TE=HE\oplus VE$. We call this smooth sub-bundle $HE$ the horizontal bundle or Ehresmann connection.
\end{Definition}

Thus, vector fields will be called \textit{vertical} or \textit{horizontal} depending on whether they belong to $\Gamma(VE)$ or $\Gamma(HE)$, respectively.

\subsection{Ehresmann Connection and Horizontal Lift}
\textls[-25]{We recall the case of the linear connection $\nabla$; it was uniquely determined by a~parallel transport~procedure.}

In the case of a~principal connection, we have an analogous.

\begin{Definition}[Lift]
Let $\pi:E\to M$ be a~fiber bundle, $M$ be a~differentiable manifold, $x\in M$ and $e\in E$ such that $\pi(e)=x$.

\textls[-25]{Given a~smooth curve $\gamma:\mathbb R\to M$ such that $\gamma(0)=x$, we define a~lift of $\gamma$ through $e$ as the curve $\tilde\gamma$, satisfying}
\begin{equation}
\tilde\gamma(0)=e\quad\text{and}\quad\pi(\tilde\gamma(t))=\gamma(t) \hspace{0.3cm} \forall t.
\end{equation}

If $E$ is smooth, then a~lift is horizontal if every tangent to $\tilde\gamma$ lies in a~fiber of $HE$, namely
\begin{equation}
\dot{\tilde\gamma}(t)\in HE_{\tilde\gamma(t)}\;\, \forall t\text{.}
\end{equation}
\end{Definition}

It can be shown that an Ehresmann connection uniquely determines a horizontal lift. Here, it is the analogy with parallel transport.

\subsection{Connection Form in a~$G$-Principal Bundle}
We now focus on the case where the smooth fiber bundle is a~$G$-principal bundle with smooth action $\mathfrak P$.

Here, we need a~group $G$, that we generally take to be a~matrix Lie group. We then have the corresponding algebra $\mathfrak g$, a~matrix vector space in the present case.

The action $\mathfrak P$ defines a~map $\sigma:\mathfrak g\to \Gamma(VE)$ called the \textit{fundamental map}\footnote{It turns out that it is an isomorphism, since $\mathfrak P$ is regular.}, where at $p\in P$, for~an~element $\xi\in \mathfrak g$, it is given via the exponential map $Exp:\mathfrak g\to G$.

\begin{equation}
\sigma_p(\xi)=\dv{}{t} \mathfrak P_{e^{t\xi}}(p)\big|_{t=0}.
\end{equation}

The map is vertical because
\begin{equation}
\pi_*\sigma_p(\xi)=\dv{}{t} \pi(\mathfrak P_{e^{t\xi}}(p))\big|_{t=0}=\dv{}{t} \pi(p)=0.
\end{equation}

Thus, the vector $\sigma_p(\xi)$ is vertical and it is called the \textit{fundamental vector}.

Before proceeding, we need some Lie group theory.

\vspace{12pt}
\noindent{\textbf{Recall of Lie machinery\footnote{In this section, we take inspiration and follow \cite{OFarrill}.}:}}
Let $G$ be a~Lie group (a differentiable manifold) with $\mathfrak g$ as its Lie algebra and $\forall g,h\in G$. We define:

\begin{itemize}
\item[--] $L_g: G\to G$ and $R_g:G\to G$, such that $L_gh=gh$ and $R_gh=hg$ are the \textit{left} and \textit{right} actions, respectively;
\item[--] the adjoint map $\text{Ad}_g:G\to G$ via such left and right actions is $\text{Ad}_g\coloneqq L_g\circ R_{g^{-1}}$, namely $\text{Ad}_gh =ghg^{-1}$. It also acts on elements of the algebra $\xi\in \mathfrak g$ as $\text{Ad}_g:\mathfrak g\to \mathfrak g$ via the exponential~map\footnote{We stress that the exponential map is not an isomorphism for all Lie groups; thus, the elements generated by the exponential map belong, in general, to a~connected subgroup of the total group, which is usually homeomorphic to its simply connected double cover. More in general, the isomorphism is between a~subset of the algebra containing $0$ and a~subset of the group containing the identity. Moreover, for a~compact, connected, and simply connected Lie group, the algebra always generates the whole group via the exponential~map.}
\begin{equation}
\begin{split}
\text{Ad}_g\xi&=\dv{}{t} \big((L_g\circ R_{g^{-1}})(e^{t\xi})\big)\big|_{t=0}=\dv{}{t} (ge^{t\xi}g^{-1})\big|_{t=0}\\[4pt]
&=g\xi g^{-1}\hspace{0.1cm}\in\mathfrak g,
\end{split}
\end{equation}
where the last two equalities hold in the present case of matrix Lie groups.
This is not to be confused with the adjoint action $\text{ad}:\mathfrak g\times\mathfrak g\to \mathfrak g$, which is generated by the derivative of the adjoint map with $g=e^{t\chi}$ and $\chi\in\mathfrak g$, such that $\text{ad}_\chi\xi=[\chi,\xi]$;
\item[--] the left invariant vector fields $v\in \Gamma(TG)$ as $L_{g*} \circ v=v$, namely $v(g)=L_{g*} v(e)$; 
\item[--] the \textit{Maurer--Cartan form} is the left invariant $\mathfrak g$-valued $1$-form $\theta\in\Omega^1(G,\mathfrak g)$ defined by its values at~$g$.
\begin{equation}
\theta_g\coloneqq L_{g^{-1}*}:T_gG\to T_eG\cong\mathfrak g.
\end{equation}

For any left invariant vector field $v$, it holds $\forall g\in G$ that $\theta_g(v(g))=v(e)$. Therefore, left invariant vector fields are identified by their values over the identity thanks to the Maurer--Cartan form $\theta$. So we can state (\cite{MF}) that this identification $v(e)\mapsto v$ defines an isomorphism between the space of left invariant vector fields on $G$ and the space of vectors in $T_eG$, thus, the Lie algebra $\mathfrak g$. For matrix Lie groups, it holds that $\theta_g=g^{-1}dg$. 
\end{itemize}

By definition, the action of an element of the group on $P$ is $\mathfrak P_g:P\to P$, and therefore, it defines a~tangent map $\mathfrak P_g{_*}\hspace{-0.09cm}:TP\to TP$, for which the following Lemma holds:

\begin{Lemma}\label{lem1}
\begin{equation}
\mathfrak P_{g*} \circ\sigma(\xi)=\sigma(\text{Ad}_{g^{-1}}\xi).
\end{equation}
\end{Lemma}
\begin{proof}
At $p\in P$
\begin{equation}
\mathfrak P_{g*} \sigma_p(\xi)=\dv{}{t} \big((\mathfrak P_g\circ\mathfrak P_{e^{t\xi}})(p)\big)\big|_{t=0}=\dv{}{t} \big((\mathfrak P_g\circ\mathfrak P_{e^{t\xi}}\circ\mathfrak P_{g^{-1}}\circ\mathfrak P_g)(p)\big)\big|_{t=0}\text{,}
\end{equation}
we then use the fact that $\mathfrak P_{g}\circ\mathfrak P_{e^{t\xi}}\circ\mathfrak P_{g^{-1}}=\mathfrak P_{g^{-1}e^{t\xi}g}=\mathfrak P_{\text{Ad}{_{g^{-1}}}} e^{t\xi}$ and the identity for matrix groups $\text{Ad}_ge^{t\xi}=e^{t\text{Ad}_{g}\xi}$ to get the following:
\begin{equation}
\mathfrak P_{g*}\sigma_p(\xi)=\dv{}{t} \big(\mathfrak P_{e^{t(\text{Ad}_{g^{-1}}\xi)}}(\mathfrak P_g(p))\big)\big|_{t=0}=\sigma_{\mathfrak P_g\hspace{-0.04cm}(p)}\hspace{-0.04cm}(\text{Ad}_{g^{-1}}\xi).
\end{equation}
\end{proof}

It is time to define what we were aiming to define at the beginning of the section: the \textit{connection form}.

\begin{Definition}
Let $P$ be a~smooth $G$-principal bundle and $HE\subset TP$ be an Ehresmann connection.\\
We call the $\mathfrak g$-valued $1$-form $\omega\in\Omega^1(P,\mathfrak g)$, satisfying
\begin{equation}
\label{eq:defconnform}
\omega(v)=\begin{cases} 
			\xi \hspace{0.7cm} if \hspace{0.2cm} 				v=\sigma(\xi), \hspace{0.1cm} \xi\in\mathcal 			C^\infty(P,\mathfrak g)\\
			
			0 \hspace{0.7cm} if \hspace{0.2cm}
			v \quad\text{horizontal},
			\end{cases}
\end{equation}
the connection $1$-form.
\end{Definition}

\begin{Proposition}
\begin{equation}
\mathfrak P_g^*\omega=\text{Ad}_{g^{-1}}\circ\omega.
\end{equation}
\end{Proposition}
\begin{proof}
Suppose $v=\sigma(\xi)$, since the other case left is trivial.

We can carry out some calculations on the left-hand side, and following from the definition of pull-back and Lemma \ref{lem1}, we have
\begin{equation}
\big(\mathfrak P_g^*\omega\big)\big(\sigma(\xi)\big)=\omega\big(\mathfrak P_{g*}\circ\sigma(\xi)\big)=\omega\big(\sigma(\text{Ad}_{g^{-1}}(\xi)\big)=\text{Ad}_{g^{-1}}(\xi).
\end{equation}

Then, we only need to manipulate the right-hand side as
\begin{equation}
\text{Ad}_{g^{-1}}\big(\omega(\sigma(\xi))\big)=\text{Ad}_{g^{-1}}(\xi).\end{equation}

Both times, we used just the given definition of connection $1$-form (Equation \eqref{eq:defconnform}).
\end{proof}

\begin{Remark}
This last Proposition is called \textit{$G$-equivariance}. It can be imposed instead of by assuming that $HE$ is an~Ehresmann connection, and then $HE$ can be shown to be such an Ehresmann connection. 
\end{Remark}

Another fundamental concept is given in the following:

\begin{Definition}[Tensorial form]
Let $\rho:G\to \operatorname{Aut}(V)$ be a~representation over a~vector space $V$ and $\alpha\in\Omega^k(P,V)$ be a~vector valued differential form.

We call $\alpha$ a~tensorial form if it is the following:
\begin{itemize}
\item[--] horizontal, i.e., $\alpha(v_1,...,v_k)=0$ if at least one $v_i$ is a~vertical vector field, and
\item[--] equivariant, i.e., for all $g\in G$, $\mathfrak P^*_g\alpha=\rho({g^{-1}})\circ \alpha$.
\end{itemize}

We define horizontal and equivariant forms as maps belonging to $\Omega^k_G(P,V)$. 
\end{Definition}

\noindent{\textbf{Observation 5:}} The connection form $\omega$ is not, in general, horizontal; thus, it is not a~tensorial form, $\omega\notin\Omega^1_G(P,\mathfrak g)$. This will be clear when taking into account how the gauge field transforms under a~change of trivialization in {Section} \ref{sec4}.

\subsection{Curvature Forms}
Given our connection $1$-form $\omega$, we can proceed in two ways: the first consists in taking a~map called the \textit{horizontal projection} and in defining the curvature as this projection applied on the exterior derivative of $\omega$. In this way, we naturally see that curvature measures the displacement of the commutator of two vectors from being horizontal.

\textls[-20]{We will proceed in a~different way though. We will define the curvature through a~\textit{structure equation}.}

\begin{Definition}
Given $\omega\in\Omega^1(P,\mathfrak g)$, a~principal connection $1$-form, the $2$-form $\Omega \in \Omega_G^2(P, \mathfrak g)$ satisfies the following:
\begin{equation}
\label{eq:defcurvform}
\Omega=d\omega+\frac{1}{2}[\omega\wedge\omega]\end{equation}
whic is called curvature $2$-form.
\end{Definition}

In Equation \eqref{eq:defcurvform}, $[\omega\wedge\omega]$ denotes the bilinear operation on the Lie algebra $\mathfrak g$ called differential Lie bracket. It is defined as follows:
\begin{equation}
[\omega\wedge\eta](u,v)=\frac{1}{2}\big([\omega(u),\eta(v)]-[\omega(v),\eta(u)]\big),
\end{equation}
where $u$ and $v$ are vector fields.

It follows straightforwardly that, if we take two general horizontal vector fields $u,v\in\Gamma(HE)$ and we use the ordinary formula\footnote{Here, we regard $\omega(u)$ as a~function $\omega(u):P\to \mathfrak g$ belonging to the algebra of smooth functions to $\mathfrak g$, $\mathcal C^\infty(P,\mathfrak g)$.} for the exterior derivative of a~$1$-form $d\omega(u,v)=u\omega(v)-v\omega(u)-\omega([u,v])$, since $\omega(u)=\omega(v)=0$, we get
\begin{equation}
\Omega(u,v)=-\omega([u,v]).
\end{equation}

We see that $\Omega$ measures how the commutator of two horizontal vector fields is far from being horizontal as well. 

\section{Exterior Covariant Derivative}\label{sec4}

\subsection{For an Ehresmann Connection $HE$}

\noindent{\textbf{Observation 6:}} $\Omega^k_G(P,V)$ is not closed under the ordinary exterior derivative. In that sense, if $\alpha\in\Omega^k_G(P,V)$, then $d\alpha\notin\Omega^{k+1}_G(P,V)$. This is what a~covariant differentiation will do instead. \\

The idea of a~covariant exterior derivative for a~connection $HE$ is, given such an Ehresmann connection $HE$, the one of projecting vector fields onto this horizontal bundle and then feed our ordinary exterior derivative with such horizontal vector fields.

First of all, we define a~map acting as a~pull-back. Namely that, given a~map $h:TP\to HE$ such that, for all vertical vector fields $v$, we get $h\circ v\coloneqq hv=0$ (called the \textit{horizontal projection}), we~define the dual map $h^*:T^*P	\to HE^*$ such that, for $\alpha\in\Omega^1(P,V)$ and $V$ a~vector space, we have $h^*\circ\alpha\coloneqq h^*\alpha=\alpha\circ h$.

\begin{Definition}[$d^h$]
Let $P$ be a~$G$-principal bundle, $V$ be a~vector space, and $\alpha\in\Omega^k(P,V)$ be an equivariant form. We define the exterior covariant derivative $d^h$ as a~map $d^h:\Omega^k(P,V)\to\Omega^{k+1}_G(P,V)$ such that
\begin{equation}
\label{eq:hexteriorcovariantderivative}
d^h\alpha(v_0,...,v_k)\coloneqq h^*d\alpha(v_0,...,v_k)=d\alpha(hv_0,...,hv_k),
\end{equation}
where $v_0,...,v_k$ are vector fields.
\end{Definition}

It depends on the choice of our Ehresmann connection $HE$, which reflects onto the horizontal projection $h$; that is why we have the index ${}^h$.

\vspace{12pt}
\noindent{\textbf{Observation 7:}} We can make our covariant derivative depend only on $\omega$, if we restrict it to only forms in $\Omega^k_G(P,V)$ and if we consider the representation of the algebra induced by the derivative of $\rho$ that we denote $d\rho:\mathfrak g\to \operatorname{End}(V)$. Then, we have $d\rho\circ\omega\in \Omega^k\big(P,\operatorname{End}(V)\big)$.

\subsection{For a~Connection Form $\omega$}

\begin{Definition}[$d_\omega$]
Let $P$ be a~$G$-principal bundle, $V$ be a~vector space, and $\alpha\in\Omega^k_G(P,V)$ be a~tensorial form. We~define the exterior covariant derivative $d_\omega$ as a~map $d_\omega:\Omega^k_G(P,V)\to\Omega^{k+1}_G(P,V)$ such that\footnote{For a~general $k$-form: $$(\omega\wedge_{d\rho}\alpha)(v_1,...,v_{k+1})=\frac{1}{(1+k)!}\sum_\sigma\text{sign($\sigma$)}d\rho\big(\omega(v_{\sigma(1)})\big)\big(\alpha(v_{\sigma(2)},...,v_{\sigma(k+1)}\big).$$}
\begin{equation}
\begin{array}{ll}
d_\omega\alpha&\coloneqq d\alpha+\omega\wedge_{d\rho}\alpha\\
&\coloneqq d\alpha+d\rho\circ\omega\wedge\alpha.
\end{array}
\end{equation}
\end{Definition}

\begin{Remark}\label{rem2}
\hspace{100pt}
\begin{itemize}
\item[--] We observe that $d^2_\omega\alpha\neq0$ for a~general $\alpha\in\Omega_G^k(P,V)$, but it is easy to show that it holds\footnote{See the first Bianchi identity in Equation \eqref{eq:firstbianchi} for the proof.}
\begin{equation}
\label{eq:propd2}
d^2_\omega\alpha=\Omega\wedge_{d\rho}\alpha,
\end{equation}
Thus, for a~\textit{flat connection} such that $\Omega=0$, we have $d^2_\omega\alpha=d^2\alpha=0$.
\item[--] We have observed that $\omega\notin\Omega_G^1(P,\mathfrak g)$. Therefore, $d_\omega\omega$ is not well defined.
However, we can consider $d^h\omega\in\Omega_G^2(P,\mathfrak g)$, and this is precisely our curvature $\Omega=d\omega+\frac{1}{2}[\omega\wedge\omega]$, where the anomalous $\frac{1}{2}$ factor comes from the "non-tensoriality" of $\omega$. As a~matter of fact, there is no representation that would make the $\frac{1}{2}$ term arise if we considered $d_\omega\omega$ instead.
\item[--] The fact that $d_\omega$ is not well defined for non-tensorial forms does not mean that $\omega$ defines a~less general derivative than what $d^h$ does. As a~matter of fact, $HE$ could be defined starting from $\omega$, as we mentioned above, since $HE=\text{Ker}\,\omega$.
\end{itemize}
\end{Remark}

\section{Gauge Field and Field Strength}

\subsection{Make It Clear}

\begin{Definition}[Gauge field]
Let $P\to M$ be a~$G$-principal bundle, $G$ be a~Lie group with $\mathfrak g$ as the respective Lie algebra, $\{U_\beta\}$ be  an open cover of $M$, and $s_\beta:U_\beta\to P$ be  a~section.

We define the gauge field as the pull-back of the connection form $\omega\in\Omega^1(P,\mathfrak g)$ as
\begin{equation}
\label{eq:gaugefield}
A_\beta=s^*_\beta\omega\hspace{0.1cm}\in\Omega^1(U_\beta,\mathfrak g).
\end{equation}
\end{Definition}

We notice that, under a~change of trivialization, such a~gauge field changes via the action of the adjoint map.

In fact, we have the following:

\begin{Lemma}
The restriction of $\omega$ to $\pi^{-1}(U_\beta)$ agrees with
\begin{equation}
\label{eq:omegarestr}
\omega_\beta=\text{Ad}_{g^{-1}_\beta}\circ \pi^*A_\beta+ g_\beta^*\theta,
\end{equation}
where $g_\beta:\pi^{-1}(U_\beta)\to G$ is the map induced by the inverse of the trivialization map $\varphi_\beta$ defined in Equation \eqref{eq:loctriv}, and with $\text{Ad}_{g^{-1}_\beta}$, we intend for the adjoint map at the group element given by $g_\beta(p)^{-1}$ at a~point $p\in\pi^{-1}(U_\beta)$.
\end{Lemma}

The proof comes from the observation that Equations \eqref{eq:defconnform} and \eqref{eq:omegarestr} coincide in $\pi^{-1}(U_\beta)$ for both a horizontal (for which they are zero) and a~vertical vector field.

Thanks to this, we easily have the following:

\begin{Proposition}
Let $G$ be a~matrix Lie group. Then it holds the following transformation for a~gauge field:
\begin{equation}
\label{eq:tensorialgf}
A_\beta=g_{\beta\gamma} A_\gamma g^{-1}_{\beta\gamma}-dg_{\beta\gamma}g^{-1}_{\beta\gamma}.
\end{equation}
\end{Proposition}
\begin{proof}
Using Equations \eqref{eq:gaugefield} and \eqref{eq:omegarestr} for all $x\in U_\beta\cap U_\gamma$,
\begin{equation}
\begin{split}
A_\beta&=s^*_\beta\omega \\
&=s^*_\beta\omega_\beta=s^*_\beta\omega_\gamma \\
&=s^*_\beta\big(\text{Ad}_{g^{-1}_\gamma}\circ\pi^*A_\gamma+g^*_\gamma\theta\big)\\
 &= \text{Ad}_{g^{-1}_{\beta\gamma}}\circ A_\gamma+g^*_{\gamma\beta}\theta \hspace{1.25cm} \text{(using $g_\gamma\circ s_\beta\coloneqq g_{\beta\gamma}:U_\beta\cap U_\gamma\to G$)}\\%
&= \text{Ad}_{g_{\beta\gamma}}\circ\big(A_\gamma-g^*_{\beta\gamma}\theta\big)\hspace{1cm} \text{$ (\text{Ad}_{g_{\beta\gamma}}\circ g^*_{\beta\gamma}\theta=-g^*_{\gamma\beta}\theta$),}\\%
\end{split}
\end{equation}
which reduces to the assert for matrix Lie groups.
\end{proof}

\vspace{12pt}
\noindent{\textbf{Observations 8:}}
\begin{enumerate}
\item We observe that a~local gauge transformation of the gauge field corresponds to a~change of trivialization chart.
\item Non-tensoriality of $\omega$ was given by the fact that it is, in general, not horizontal. For the gauge field $A$, we can generalize to forms on $M$ the concept of tensoriality/non-tensoriality by~noticing that a~form obtained by the pull-back of a~tensorial form, denoted with $t\in\Omega^1_G(P,V)$, would~transform differently compared to $A$, namely as 
\begin{equation}
\label{eq:horform}
t_\beta\coloneqq s_\beta^*t=g_{\beta\gamma} t_\gamma g^{-1}_{\beta\gamma}.
\end{equation}

The Maurer--Cartan form $\theta$ reflects the non-horizontality of $\omega$ to the gauge field, from Equation~\eqref{eq:omegarestr}.
\item A difference of two gauge fields like $A-A'$ transforms as Equation \eqref{eq:horform}. In fact, the transformation rule is one of a~tensorial form, since the Maurer--Cartan forms simplify.
\item We notice that (iii) is a~particular case of a~more general one. Indeed, it is possible to show with proof in Reference \cite{KN} (Chapter $5$) that $\Omega^k_G(P,V)\cong\Omega^k(M,P\times_\rho V)$. This is essentially due to the fact that, thanks to the equivalence relation of the associated bundle and the gluing condition of sections on overlaps, the pull-backs by sections $s_\beta:U_\beta\to P$ give a~one-to-one correspondence between these two spaces. Therefore, we can obtain forms with a~tensorial transformation like Equation \eqref{eq:horform} just by taking the pull-back of tensorial forms on $P$; these will be forms on $M$ with values into the associated bundle $P\times_\rho V$.
\item Observations {(iii)} and {(iv)} ensure that an object built with gauge fields $A_\beta\in\Omega^1(U_\beta,\mathfrak g)$ (which transform on overlaps by Equation \eqref{eq:tensorialgf}) will be in $\Omega^2(M, P\times_{\text{Ad}} \mathfrak g)$; see Observation $9$.
\end{enumerate}

\begin{Definition}[Field strength]
Let $P\to M$ be a~$G$-principal bundle, $G$ be a~Lie group with $\mathfrak g$ as the respective Lie algebra, $\{U_\beta\}$ be an open cover of $M$, and $s_\beta:U_\beta\to P$ be a~section.

We define the field strength as the pull-back of the curvature form $\Omega\in\Omega^2_G(P,\mathfrak g)$ as
\begin{equation}
F_\beta=s^*_\beta\Omega\hspace{0.1cm}\in\Omega^2_G(U_\beta,\mathfrak g),
\end{equation}
which, by definition of $\Omega$, is
\begin{equation}
F_\beta=dA_\beta+\frac{1}{2}[A_\beta\wedge A_\beta].
\end{equation}
\end{Definition}

Similarly to what we have done for the gauge field, we can show\footnote{Using the Cartan structure equation for $\theta$, $d\theta=-\frac{1}{2}[\theta,\theta]$.} that the field strength transforms as
\begin{equation}
\label{eq:tranffieldstrength}
F_\beta=\text{Ad}_{g_{\beta\gamma}}\circ F_\gamma=g_{\beta\gamma} F_\gamma  g_{\beta\gamma}^{-1},
\end{equation}
where the last equality holds for matrix Lie groups with $g$ and $g^{-1}$ in $G$. This is indeed the transformation of a~tensorial form, as in Equation \eqref{eq:horform}.

\vspace{10pt}
\noindent{\textbf{Observation 9:}} Thanks to our previous observation, i.e., there is a~canonical isomorphism between $\Omega^k_G(P,V)$ and $\Omega^k(M,P\times_\rho V)$, we can relate $\Omega$ and $F_\beta$ with a~form\footnote{Where we have introduced the notation $\Omega^k(M,P\times_{\text{Ad}}\mathfrak g)\coloneqq\Omega^k(M,\text{ad}P)$.} $F_A\in\Omega^2(M,\text{ad}P)$. Namely~there is a~canonical isomorphism sending $\Omega\in\Omega^2_G(P,\mathfrak g)$ to $F_A\in\Omega^2(M,\text{ad}P)$. Indeed, given the transformation law for the field strength in Equation \eqref{eq:tranffieldstrength}, we see that $\{F_\beta\}$ are horizontal and equivariant and, thus, form a~global section belonging to $\Omega^2(M,\text{ad}P)$, which is usually denoted as $F_A$. \\
The notation $F_A$ stresses that it is obtained from gauge fields in $\Omega^1(U_\beta,\mathfrak g)$.\\
In the case of a~trivial bundle, it is also possible to define a~global gauge field $A\in\Omega^1(M,\mathfrak g)$.

\subsection{2{nd} Bianchi Identity}
\begin{Definition} The collection of gauge fields defines an exterior covariant derivative for bundle-valued forms on $M$. We denote such a map with
\begin{equation}
\label{eq:derivativebundlevalued}
d_A:\Omega^{k}(M,P\times_\rho V)\to\Omega^{k+1}(M,P\times_\rho V)\text{.}
\end{equation}
\end{Definition}
Consider $d_A:\Omega^k(M,P\times_\rho V)\to\Omega^{k+1}(M,P\times_\rho V)$ as the exterior covariant derivative and $F_A\in\Omega^2(M,\operatorname{ad}P)$ as the field strength.

Then, we have the following:
\begin{Proposition}
\begin{equation}
\label{eq:secondbianchi}
d_A F_A=0.
\end{equation}
This is the second Bianchi identity.
\end{Proposition}

\begin{proof}
Given
\begin{equation}
F_A=dA+\frac{1}{2}[A\wedge A],
\end{equation}
then
\begin{equation}
\begin{split}
d_AF_A&=dF_A+[A\wedge F_A]\\
&=d^2A+\frac{1}{2}d[A\wedge A]+[A\wedge dA]+\frac{1}{2}[A\wedge[A\wedge A]]\\
&=\frac{1}{2}[A\wedge[A\wedge A]] \hspace{2cm} \text{($d^2A=0$ and $\frac{1}{2}d[A\wedge A]=-[A\wedge dA]$)}\\%
&=0.\hspace{4.021cm}\text{(because of Jacobi identity)}\\%
\end{split}
\end{equation}
\end{proof}

\section{Affine Formulation}
In the usual formulation of General Relativity, one defines a~covariant derivative $\nabla$, which is a~map among tensors. Then, one can define curvature and torsion and eventually get the field equations for ECSK theory or General Relativity by setting torsion to zero.

One may wonder if this latter formulation is equivalent to the one we have been implementing through principal bundles and principal connection.

The answer is positive and is given in the next two sections.
\subsection{Affine Covariant Derivative}\label{sec6dot1}
We have built our setup by taking $\rho$ to be the fundamental representation of $\text{O}(3,1)$, $P=F_O(M)$, and $\mathcal V=F_O(M)\times_{\rho}V$ to be the Minkowski bundle, as in $(ii)$ of {Observations} $4$.\\
We are now ready to have a view of the problem from another perspective.
\begin{Definition}
Let $\{U_\beta\}$ be an open cover of $M$. We define a local connection as a $\Lambda^2\mathcal V$-valued differential form $\omega_\beta\in\Omega^1(U_\beta,\Lambda^2\mathcal V)$. Then we define the space of local connections as $$\mathcal A_M \coloneqq\big\{ \omega_\beta\in\Omega^1(U_\beta,\Lambda^2\mathcal V)\big\}.$$
\end{Definition}
In literature these local connections are also called \textit{spin connections}.\\
We come now to the following Proposition.
\begin{Proposition}
There exists an isomorphism
\begin{equation}
\label{eq:isoalgebrawedge2}
\Lambda^2 V\xrightarrow{\sim} \mathfrak{so}(3,1)
\end{equation}
induced by the reference metric $\eta$. 
\end{Proposition}
\begin{proof}
Given $\eta=\text{diag}(1,1,1,-1)$, consider a basis of $V$ given by $\{v_i\}_{i=1}^4$. A basis for $\Lambda^2 V$ is obtained by taking the wedge product. Then the proof is an immediate consequence of the fact that the elements $L_i{^j}$ induced by the map $\eta:\Lambda^2 V\to\mathfrak{so}(3,1)$ via
\begin{equation}
L_i{^j} = v_i\wedge v_k\eta^{jk} 
\end{equation}
correspond to a basis of the Lie algebra.
\end{proof}

This isomorphism allows to obtain a different approach to the bundle-valued derivative defined in Equation \eqref{eq:derivativebundlevalued} by means of such local connections. In fact, instead of the collection of gauge fields, one can consider the space of local connections $\mathcal A_M$.
\begin{Definition}
Let $\rho:\operatorname{O}(3,1)\to\operatorname{Aut}(V)$ be the fundamental representation. Then the space of local connections $\mathcal A_M$ and the isomorphism \eqref{eq:isoalgebrawedge2} allow to define a bundle-valued exterior covariant derivative
\begin{equation}
\label{eq:ultimatederivative}
d_\omega:\Omega^k(M,\mathcal V)\to\Omega^{k+1}(M,\mathcal V).
\end{equation}
\end{Definition}

Notice that this is denoted in the very same way of the exterior covariant derivative for differential forms on principal bundles. However, the distinction will be always evident.

We can further define another kind of derivative that ``takes care'' of internal indices only; in particular, this will not be necessarily a~map between differential forms. 

\begin{Definition}
Given $\omega\in\mathcal A_M$ and $\phi \in \Gamma(\mathcal V)$ restricted to the same neighbourhood, we define $(D_\omega)_\mu$ as
\begin{equation}
(D_\omega\phi)_\mu^a=(\partial_\mu\phi^a+\omega_\mu^{ac}\eta_{cb}\phi^b)
\end{equation}

and, for $\alpha\in\Omega^k(M,\mathcal V)$, we have,
\begin{equation}
\label{eq:Domegaforgeneralform}
(D_\omega\alpha)_{\mu\nu_1...\nu_k}^a=(\partial_\mu\alpha_{\nu_1...\nu_k}^a+\omega_\mu^{ac}\eta_{cb}\alpha_{\nu_1...\nu_k}^b).
\end{equation}
\end{Definition}
Equation \eqref{eq:Domegaforgeneralform} shows that it does not map $\alpha$ to a~differential form.\\
The antisymmetry of $\omega$ ensures the metric compatibility condition for this derivative. In fact, it is easy to check that $(D_\omega\eta)_\mu^{ab}=0$ for each $\mu, a, b$.

Now, we immediately apply the inverse of a~tetrad to $D_\omega\phi$ and identify it with $\nabla$. In fact, we take a~vector field $X\in \Gamma(TM)$, feed the tetrad $e$ with it, then apply\footnote{Here, we use the so-called \textit{interior product}, i.e., a~map $\iota_\xi:\Omega^{k}(M)\to\Omega^{k-1}(M)$, such that $(\iota_\xi\alpha)(X_1,...,X_{k-1})=\alpha(\xi,X_1,...,X_{k-1})$, for vector fields $\xi, X_1,...X_{k-1}$. Furthermore it respects $\iota_\xi(\alpha\wedge\beta)=(\iota_\xi\alpha)\wedge\beta+(-1)^k\alpha\wedge(\iota_\xi\beta)$, where~$\alpha\in\Omega^k(M)$. Therefore, it forms an antiderivation. The relation with the Lie derivative is given by the formula $\mathcal L_\xi\alpha=d(\iota_\xi\alpha)+\iota_\xi d\alpha$, called the Cartan identity. The interior product of a~commutator satisfies $\iota_{[X,Y]}=[\mathcal L_X, \iota_Y]$, with $X$ and $Y$ vector~fields. } $D_\omega$ to get $D_\omega (\iota_Xe)$, and finally pull it back with the inverse of the tetrad $\bar e$.

In components, this reads as follows:
\begin{equation}
\label{eq:Dontetrads}
\big(D_\omega(\iota_Xe)\big)^a_\mu=D_\mu(e_\nu^aX^\nu)=\partial_\mu (e^a_\nu X^\nu)+\omega^{ab}_\mu\eta_{bc}e^c_\nu X^\nu.
\end{equation}

Pulling back via $\bar e$, we obtain
\begin{equation}
\bar e_a^\sigma\big( D_\mu(e_\nu^aX^\nu)\big)=\bar e_a^\sigma\big(\partial_\mu (e^a_\nu X^\nu)+\omega^{ab}_\mu\eta_{bc}e^c_\nu X^\nu\big).
\end{equation}

We define the Christoffel symbols $\Gamma_{\mu\nu}^\sigma$ as
\begin{equation}
\begin{split}
\Gamma_{\mu\nu}^\sigma&=\bar e^\sigma_a(D_\mu e_\nu^a)\\
&=\bar e^\sigma_a (\partial_\mu e^a_\nu+\omega^{ab}_\mu\eta_{bc} e^c_\nu)
\end{split}
\end{equation}
and, thus, we get
\begin{equation}
\label{eq:nabladef}
\nabla_\mu X^\sigma\coloneqq\bar e_a^\sigma\big( D_\mu(e_\nu^aX^\nu)\big)=\partial_\mu X^\sigma+\Gamma_{\mu\nu}^\sigma X^\nu,
\end{equation}
which is the covariant derivative well known in General Relativity.

We can also see what the curvature form is in terms of the commutator of two of these derivatives, $F_\omega$ and it is given by
\begin{equation}
\big(D_{[\mu}D_{\nu]}\phi\big)^a=\big(\partial_{[\mu}\omega{_{\nu]}}^{ab}+\omega_{[\mu}^{ad}\eta_{de}\omega_{\nu]}^{eb}\big)\eta_{bc}\phi^c=F_{\mu\nu}^{ab}\eta_{bc}\phi^c,
\end{equation}
where $A_{[\mu}B_{\nu]}=A_{\mu}B_{\nu}-A_{\nu}B_{\mu}$ is our convention for the antisymmetrization.
The fact that $F_\omega$ is a~$2$-form shows that $F^{ab}_{\mu\nu}=-F^{ab}_{\nu\mu}$; furthermore, it also holds $F^{ab}_{\mu\nu}=-F^{ba}_{\mu\nu}$.

\subsection{Riemann Curvature Tensor}

\textls[-25]{We can now consider the commutator of two affine covariant derivatives and use Equation \eqref{eq:nabladef}, getting}
\begin{equation}
(\nabla_{[\mu}\nabla_{\nu]}X)^\sigma=\bar e^\sigma_a\big(D_{[\mu} D_{\nu]}( \iota_X e)\big)^a=\bar e^\sigma_a F^{ab}_{\mu\nu}\eta_{bc} e^c_\omega X^\omega.
\end{equation}

We identify the Riemann tensor
\begin{equation}
R_{\mu\nu\omega}{^\sigma}=\bar e^\sigma_a F_{\mu\nu}^{ab}\eta_{bc}e^c_\omega,
\end{equation}
the Ricci curvature tensor
\begin{equation}
R_{\mu\omega}=R_{\mu\sigma\omega}{^\sigma}=\bar e^\sigma_a F^{ab}_{\mu\sigma}\eta_{bc}e^c_\omega
\end{equation}
and thus the Ricci scalar
\begin{equation}
R=g^{\mu\omega}R_{\mu\omega}=\bar e^\mu_d\bar e^\omega_e\eta^{de}\bar e^\sigma_a F^{ab}_{\mu\sigma}\eta_{bc}e^c_\omega=-\bar e^\mu_a\bar e^\omega_b F^{ab}_{\mu\omega}.
\end{equation}

It follows the antisymmetry of the Riemann tensor in the indices $\mu\nu$ and $\omega\sigma$, but it is important to note that we cannot ensure any symmetry in the Ricci curvature instead due to the presence of torsion.

\section{Torsion}\label{sec7}
Here, we start focusing on the importance of torsion, which arises quite naturally as curvature does.

\subsection{Torsion Form}

\begin{Definition}[Solder form/soldering of a~$G$-principal bundle]
Let $\pi: P\to M$ be a~smooth $G$-principal bundle over a~differentiable manifold $M$, $\rho:G\to\operatorname{Aut}(V)$ be a~representation, and $G$ be a~Lie group.

We define the solder form, or soldering, as the vector-valued $1$-form $\theta\in\Omega^1_G(P,V)$ such that $\tilde\theta:TM\to P\times_\rho V$ is a~bundle isomorphism, where $\tilde\theta\in\Omega^1(M,P\times_\rho V)$ is the associated bundle map induced by the isomorphism of $\Omega^1_G(P,V)\cong\Omega^1(M,P\times_\rho V)$.
\end{Definition}

\noindent{\textbf{Observations 10:}}
\begin{itemize}
\item[--] The choice of the solder form is not unique, in general.
\item[--] We can observe that, taking $P=F_O(M)$, $\rho$ as the fundamental representation of O$(3,1)$, and $V$ as the vector space with reference metric $\eta$, $\tilde\theta$ corresponds to our definition of tetrads. The different choices of soldering give rise to different tetrads.
\item[--] In the case that $P=F_O(M)$ and that the associated bundle is simply chosen to be $TM$, the solder form is called \textit{canonical} or \textit{tautological}. Since the associated bundle $TM$ sets the bundle isomorphism $\tilde\theta$ to be the identity map $\text{id}:TM\to TM$. 
\item[--] In {Observations} $4$, we mentioned that the Minkowski bundle cannot be canonically identified with the tangent bundle itself; indeed, we fixed a~reference metric $\eta$, which cannot be pulled back by the identity map to give the metric on $TM$ in general, and thus, the solder form is not canonical.
\end{itemize}

The soldering of the principal frame bundle allows us to define the torsion form\footnote{Torsion can be defined for every principal bundle, but physics arises when considering the frame bundle.}.

\begin{Definition}[Torsion form]
Let $P=F_O(M)$, $\rho:\text{O}(3,1)\to\operatorname{Aut}(V)$ be the fundamental representation, $V$ be a~vector space with reference metric $\eta$, and $\theta\in\Omega^1_G(P,V)$ be a~solder form.\\
We define the torsion form $\Theta\in\Omega^2_G(P,V)$ as follows:
\begin{equation}
\label{eq:torsionform}
\Theta=d_\omega\theta=d\theta+\omega\wedge_f\theta.
\end{equation}
\end{Definition}

\subsection{Torsion in a~Local Basis}
We would like to express the torsion form in terms of tetrads and the gauge field.

In Reference \cite{GS}, a~formula is given and it is obtained by applying the previous definition of the torsion form in a local basis
\begin{equation}
\tilde\Theta^a=(d_\omega e)^a=de^a+\omega^{ab}\eta_{bc}\wedge e^c,
\end{equation}
where $\omega\in \mathcal A_M$ is a local connection.

\subsection{1st Bianchi Identity}
\begin{Proposition}
Following our previous definitions, we have
\begin{equation}
\label{eq:firstbianchi}
d_\omega\Theta=\Omega\wedge_f\theta,
\end{equation}
which is called the first Bianchi identity.
\end{Proposition}
\begin{proof}
For this proof, we prefer using Equation \eqref{eq:hexteriorcovariantderivative}.

We consider three vector fields $u,v,w\in\Gamma(TP)$. By definition, it follows

\begin{equation}
\begin{split}
d^h\Theta(u,v,w)&=d\Theta(hu,hv,hw)\\[3pt]
&=(d\omega\wedge_f\theta-\omega\wedge_fd\theta)(hu,hv,hw)\quad\text{(because of Equation \eqref{eq:torsionform})}\\[3pt]
&=d\omega\wedge_f\theta(hu,hv,hw)\quad\text{(because of Equation \eqref{eq:defconnform})}\\[3pt]
&=\Omega\wedge_f\theta(u,v,w)\text{.}
\end{split}
\end{equation}

The last equality holds because of tensoriality of $\theta$ and the second remark in {Remark} \ref{rem2}.
\end{proof}
This proposition is a~natural consequence of the property of the covariant differential expressed in Equation \eqref{eq:propd2}.

\subsection{Torsion Tensor}
\begin{Definition}[Torsion tensor]
Given two vector fields $X,Y\in \Gamma(TM)$ and a~$1$-form $\tau\in\Omega^1(M)$, we define the torsion tensor field $Q$ as the tensor field of type-$\binom{1}{2}$ such that
\begin{equation}
Q(X,Y;\tau)\coloneqq \tau(Q(X,Y))=\tau\big(\bar e(d_\omega e(X, Y))\big).
\end{equation} 

It is evidently antisymmetric in $X, Y$, by definition. 
\end{Definition}

\begin{Proposition}
We have the following formula:
\begin{equation}
Q(X,Y)=\nabla_XY-\nabla_YX-[X,Y]
\end{equation}
and, in components, it reads
\begin{equation}
Q_{\mu\nu}{^\sigma}=\Gamma^\sigma_{\mu\nu}-\Gamma^\sigma_{\nu\mu}-C^\sigma_{\mu\nu},
\end{equation}
where $C^\sigma_{\mu\nu}=0$ in a~holonomic basis for $X$ and $Y$ and $\nabla$ is the covariant derivative\footnote{See Reference \cite{strau} for references about this.}.
\end{Proposition}
\begin{proof}
Working in a local basis and recalling the definition of torsion
\begin{equation}
Q=\bar e\cdot(d_\omega e)=\bar e_a(d_\omega e)^a,
\end{equation}
it follows
\begin{equation}
\begin{array}{ll}
\bar e_a(d_\omega e)^a&=\bar e^\sigma_a\big(\partial_{[\mu} e^a{_{\nu]}}+\omega^a_{[\mu b}e^b_{\nu]}\big)dx^\mu\wedge dx^\nu\otimes\partial_\sigma\\
&=(\Gamma^\sigma_{\mu\nu}-\Gamma^\sigma_{\nu\mu})dx^\mu\wedge dx^\nu\otimes\partial_\sigma \quad \text{($\Gamma^\sigma_{\mu\nu}=\bar e^\sigma_a(D_\mu e^a_\nu)$.)}\\
&=Q_{\mu\nu}{^\sigma}dx^\mu\wedge dx^\nu\otimes\partial_\sigma\text{,}
\end{array}
\end{equation}
then $Q_{\mu\nu}{^\sigma}=\Gamma^\sigma_{\mu\nu}-\Gamma^\sigma_{\nu\mu}$.
\end{proof}

We have now set up all the background for building our theory and for discussing field equations of ECSK theory.

\section{Field Equations and Conservation Laws}
We present here field equations for ECSK theory\footnote{Some classical works about ECSK theory and General Relativity with torsion, like References \cite{Hehl:2007bn,RevModPhys.48.393,Chakrabarty:2018ybk}.}.  Thus, we will neither assume the possibility of a~propagating torsion (and we will always keep non-identically vanishing Riemann curvature \cite{Nester:1998mp}) nor display a~lagrangian for a~totally independent torsion field; rather, we will only set the Palatini--Cartan lagrangian for gravity, as done in Reference \cite{Cattaneo:2017hus}, and a~matter lagrangian as the source. This theory is known as Einstein--Cartan--Sciama--Kibble gravity (ECSK).

\textls[-20]{In the present case, torsion reduces to an algebraic constraint. This is a~consequence of making torsion join the action of the theory as only contained in the Ricci scalar because of a~non-torsion-free connection and not with an independent coupling coefficient. In works like References \cite{Fabbri:2017rjf,Fabbri:2018hhv,Fabbri:2017xlx,Fabbri:2014dxa}, torsion is present as an~independent part (independent coupling coefficient) of the action and it does propagate.}

This is why the ECSK is considered as the most immediate generalization of General Relativity with the presence of torsion.

Therefore, we wish to eventually obtain an action of two independent objects, tetrads and connection, where this latter action should give rise to equations for curvature when varying tetrads and for torsion when varying the connection.

We will focus more on the geometrical side of these equations and we will not dwell on deepening matter interaction (couplings, symmetry breaking, etc.), as done for instance in References \cite{FQ,Gies:2013noa,Diether:2019nxe,Cabral:2019gzh,Inglis:2017tmu}.

\subsection{ECSK Equations}
ECSK theory with cosmological constant belongs to the Lovelock--Cartan family, which describes the most general action in four dimensions such that this action is a~polynomial on the tetrads and the spin connection (including derivatives), is invariant under diffeomorphisms and local Lorentz transformations, and is constructed without
the Hodge dual\footnote{See Reference \cite{mardones} for details.}.

The variational problem is given by an action of the kind
\begin{equation}
\label{eq:action}
S=S_{PC}+S_\textit{matter},
\end{equation}
where the Palatini--Cartan action is
\begin{equation}
\label{eq:geoaction}
S_{PC}[e,\omega]=\int_M\operatorname{Tr}\big[\frac{1}{2}e\wedge e \wedge F_\omega+\frac{\Lambda}{4!}e^4\big].
\end{equation}

We work in a system of local connections in $\mathcal A_M$. The wedge product is defined over both space--time and internal indices as a~map\footnote{Such that, for $\alpha\in\Omega^k(M,\Lambda^p\mathcal V)$ and $\beta\in\Omega^l(M,\Lambda^q\mathcal V)$, we have $\alpha\wedge\beta=(-1)^{(k+p)(l+q)}\beta\wedge\alpha$.}$\wedge:\Omega^k(M,\Lambda^p\mathcal V)\times\Omega^l(M,\Lambda^q\mathcal V)\to\Omega^{k+l}(M,\Lambda^{p+q}\mathcal V)$ and the trace is a~map $\text{Tr}:\Lambda^4 V\to \mathbb R$, normalized such that (for $v_i$ elements of a~basis in $V$) $\text{Tr}[v_i\wedge v_j\wedge v_k\wedge v_l]=\varepsilon_{ijkl}$. The choice of the normalization of the trace works as a~choice of orientation for $M$ (since the determinant of a~matrix in $O(3,1)$ may be $\pm1$). Therefore, we reduce the total improper Lorentz group O$(3,1)$ to the only orientation preserving part, which is still not connected, SO$(3,1)$. This gives an invariant volume form on $M$. In this way, we consider sections of $\Lambda^kT^*M\otimes\Lambda^p\mathcal V$.

Later on, we will make explicit some indices and keep implicit some others; for this purpose, we will specify what kind of wedge product we are dealing with, even though it will be evident because it will be among the implicit indices.

We recall the definition of $F_\omega$ and deduce the identity for its variation
\begin{equation}
\label{eq:omegavariation}
\delta_\omega F_\omega=d_\omega\delta\omega,
\end{equation}
where we stress that, despite $\omega$ being non-tensorial, $\delta\omega$ is instead, and thus $d_\omega\delta\omega$ is well defined.

The action for the matter is of the kind
\begin{equation}
\label{eq:mataction}
S_{matter}{[e,\omega,\varphi]}=\kappa\int_M\text{Tr}[L(e,\omega, \varphi)],
\end{equation}
where $L$ is an invariant lagrangian density form with the proper derivative order in our variables, $\varphi$ is a~matter field, and $\kappa$ is a~constant.

Such matter lagrangian is supposed to be source for both curvature and torsion equations, namely~it will be set for fulfilling some conditions fitting the theory.

Therefore, varying the actions in Equations \eqref{eq:geoaction} and \eqref{eq:mataction} and considering Equation \eqref{eq:omegavariation}, we have\footnote{Omitting equations of motion $\frac{\delta L}{\delta\varphi}=0$ for the matter field, which have to be satisfied for conservation laws anyway.} 

\begin{equation}
\label{eq:intFE}
\begin{array}{ll}
\int_M\text{Tr}[\delta e\wedge( e\wedge F_\omega+\frac{\Lambda}{3!} e^3)] &=\int_M \text{Tr}[\kappa\frac{\delta L}{\delta e}\wedge\delta e]\\[3.5pt]
\int_M\text{Tr}[\frac{1}{2}d_\omega(e\wedge e)\wedge\delta\omega ]&=\int_M\text{Tr}[\kappa \frac{\delta L}{\delta \omega}\wedge\delta\omega]\text{,}\\
\end{array}
\end{equation}
which is equivalent to
\begin{equation}
\label{eq:FE}
\begin{array}{ll}
\varepsilon_{abcd}e^b\wedge F^{cd}_\omega+\frac{\Lambda}{3!} \varepsilon_{abcd} e^b\wedge e^c\wedge e^d &= \kappa\frac{\delta\text{Tr}[L]}{\delta e^a}\coloneqq \kappa T_a\\[3pt]
\frac{1}{2}\varepsilon_{abcd}d_\omega(e^c\wedge e^d) &=\kappa \frac{\delta\text{Tr}[L]}{\delta \omega^{ab}}\coloneqq\kappa\Sigma_{ab}\\
\end{array}
\end{equation}
where the wedge product here is only between differential forms.

Setting $\Lambda=0$ and in performing the derivative, Equation \eqref{eq:FE} can be rewritten as
\begin{equation}
\label{eq:NFE}
\begin{array} {ll}
\varepsilon_{abcd}e^b\wedge F^{cd}_\omega &= \kappa T_a\\
\varepsilon_{abcd}\,\tilde{Q}^c\wedge e^d &=\kappa\Sigma_{ab}\text{,}\\
\end{array}
\end{equation}
where we have set $\tilde{Q}=d_\omega e$.

These are equations for the ECSK theory in their implicit form\footnote{Without making space--time indices explicit.}, where $T$ and $\Sigma$ are related to, respectively, the \textit{energy momentum} and the \textit{spin} tensor, once pulled back.

By making all the indices explicit, as given in Reference \cite{FQ}, and properly setting $\kappa$ according to natural units\footnote{All fundamental constants $=1$.}, Equation \eqref{eq:NFE} takes the following form
\begin{equation}
\begin{array}{ll}
G_{\mu\nu}&=8\pi T_{\mu\nu}\\[3pt]
Q_{\mu\nu}{^\sigma}&=-16\pi\Sigma_{\mu\nu}{^\sigma}\text{.}\\
\end{array}
\end{equation}

\vspace{12pt}
\noindent{\textbf{Observations 11:}}
\begin{enumerate}
\item $T_{\mu\nu}$ is not symmetric, as expected from the non-symmetry of the Ricci curvature $R_{\mu\nu}$.
\item We stress that, even though $e$ is an isomorphism, the map $e\wedge\boldsymbol{\cdot}:\Omega^k(M,\Lambda^p\mathcal V)\to\Omega^{k+1}(M,\Lambda^{p+1}\mathcal V)$ is not an isomorphism, in general. In fact, taking $\frac{\delta L}{\delta e}=0$ (with $\Lambda=0$) in Equation \eqref{eq:intFE} does not imply $F_\omega=0$, which would imply a~flat connection.
\item Setting $\frac{\delta L}{\delta \omega}=0$ in Equation \eqref{eq:intFE}, one recovers the condition of vanishing torsion (hence, a Levi--Civita connection) and, therefore, the Einstein equations.
\item It is interesting to note that, requiring a~totally antisymmetric spin tensor, sets the total antisymmetry of the torsion tensor. Namely, in the case of a~totally antisymmetric $\Sigma$, we need to couple the only totally antisymmetric part of torsion into the geometrical lagrangian. This is further discussed in Reference \cite{FQ}.
\end{enumerate}

\subsection{Conservation Laws}

We have two symmetries, i.e., local Lorentz transformations and diffeomorphisms. They are continuous symmetries, and as such, we expect two conservation laws. Since we are dealing with local symmetries, we shall not find two conserved currents but rather two relations for the variations of the matter lagrangian w.r.t. $e$ and $\omega$. 

These relations directly imply the Bianchi identities of Equations \eqref{eq:secondbianchi} and \eqref{eq:firstbianchi}, but we could also do the converse, namely assuming Equations \eqref{eq:secondbianchi} and \eqref{eq:firstbianchi} and then deriving such conservation laws. This means that conservation laws are a~consequence of symmetry on the one hand, implemented via the following symmetries (respectively diffeomorphisms and local SO$(3,1)$)

\begin{equation}
\label{eq:diffeotransf}
\begin{array} {ll}
\delta_\xi e^a&=\mathcal L_\xi e^a=\iota_\xi de^a+d\iota_\xi e^a\\[2pt]
\delta_\xi \omega^{ab}&=\mathcal L_\xi \omega^{ab}=\iota_\xi d\omega^{ab}+d\iota_\xi \omega^{ab},\\
\end{array}
\end{equation}
where $\xi$ is the generator vector field,
\begin{equation}
\label{eq:lorentztransf}
\begin{array}{ll}
\delta_\Lambda e^a&=\Lambda^{ab}\eta_{bc}e^c\\[2pt]
\delta_\Lambda \omega^{ab}&=-d_\omega\Lambda^{ab}\quad\Lambda\in\mathfrak{so}\text{(3,1)},
\end{array}
\end{equation}
or a~direct consequence if we impose field equations and, thus, gravitational dynamics and Bianchi identities on the other hand.

We will follow the shortest derivation, namely to implement the Bianchi identities of Equations \eqref{eq:secondbianchi} and \eqref{eq:firstbianchi} on field Equation \eqref{eq:NFE}.

Thanks to Bianchi identities, left hand side of field Equation \eqref{eq:NFE} can be rewritten in the following way: 
\begin{equation}
\label{eq:clbianchi}
\begin{array}{ll}
d_\omega (\varepsilon_{abcd}e^b\wedge F^{cd}_\omega)&=\iota_a\tilde Q^b\wedge(\varepsilon_{bcde}e^c\wedge F^{de}_\omega) +\iota_aF_\omega^{bc}\wedge(\varepsilon_{bcde}\tilde Q^d\wedge e^e)\\[3pt]
d_\omega(\varepsilon_{abcd}\tilde Q^c\wedge e^d)&=-\frac{1}{2}(\varepsilon_{acde}e^c\wedge F^{de}_\omega\wedge e_b-\varepsilon_{bcde}e^c\wedge F^{de}_\omega\wedge e_a),
\end{array}
\end{equation}
where $\iota_a=\iota_{\bar e_a}$ and $e_b=\eta_{bc}e^c$.

However, because of the same field in Equation \eqref{eq:NFE}, they reduce to

\begin{equation}
\label{eq:conslaws}
\begin{array}{ll}
d_\omega T_a&=\iota_a\tilde Q^b\wedge T_b+\iota_aF_\omega^{bc}\wedge\Sigma_{bc}\\[2pt]
d_\omega\Sigma_{ab}&=\frac{1}{2}T_{[a}\wedge e{_{b]}},
\end{array}
\end{equation}

\textls[-25]{In References \cite{Bonder:2018mfz} and \cite{Jiang:2000xp}, a~more detailed discussion can be found.
These are conservation laws for ECSK theory.}

In components, as given in Reference \cite{FQ}, they read
\begin{equation}
\begin{array}{ll} 
\nabla_\mu T^{\mu\nu}+T_{\sigma\rho}Q^{\sigma\rho\nu}-\Sigma_{\mu\sigma\rho}R^{\mu\sigma\rho\nu}&=0\\[3pt]
\nabla_\mu\Sigma_{\sigma\omega}{^\mu}+\frac{1}{2}T_{[\sigma\omega]}&=0.
\end{array}
\end{equation}

\section{Conclusions}
We have set up all the mathematical background for building ECSK theory, eventually achieving field equations and conservation laws. 

In ECSK theory, torsion is only an algebraic constraint and it does not propagate. This is a~natural consequence of inserting torsion into the theory without an independent coupling coefficient but simply generalizing the Einstein--Hilbert action (or Palatini action in our formulation) $\int R\sqrt{-g}d^4x$ to a~non-torsion-free connection $\nabla$ (or spin connection in our case). In this case, the Ricci scalar contains both curvature and torsion.

It is possible to immediately recover General Relativity by imposing the zero torsion condition, which, in the considered theory, translates to letting the matter field $\varphi$ generate a~null contribution to the spin tensor $\Sigma_{\mu\nu}{^\sigma}$. The most natural matter fields which would fit with the theory are spinors; indeed, spinors are the way in which we can have a~non-vanishing spin tensor which is also dynamical because of equations of motion for the spinor field.

This review does not want to substitute the well-known literature but to just give a~self-contained and mathematically rigorous introduction to ECSK theory, providing also some references for deepening knowledge in the present subjects. Also, we intentionally did not dive too deeply into physical applications to cosmology (like done in References \cite{Poplawski1,Poplawski2,BravoMedina,Mehdizadeh,Pesmatsiou,Kranas}), that might be a~valid argument for writing another review article.

\end{document}